\documentclass[unsortedaddress,onecolumn, nofootinbib,10pt]{revtex4}
\usepackage[utf8]{inputenc}
\usepackage{amsmath,empheq,color}
\usepackage{amsfonts}
\usepackage{amsthm}
\usepackage{amssymb}
\usepackage{graphicx}
\usepackage{hyperref}
\usepackage[normalem]{ulem}
\usepackage{epigraph}
\usepackage{mathrsfs}
\usepackage{bbm}
\usepackage{wrapfig}
\usepackage{lineno}
\usepackage{scalerel,stackengine}
\usepackage{epigraph}
\usepackage{cancel}
\usepackage{soul}
\usepackage{ulem}

\newtheorem{theorem}{Theorem}

\newtheorem{corollary}{Corollary}
\newtheorem{example}{Example}
\newtheorem{definition}{Definition}
\newtheorem{proposition}{Proposition}
\newtheorem{lemma}{Lemma}

\renewcommand{\L}{\mathfrak{L}}
\newcommand{\X}{X_{h}^{t}}

\begin{document}

\title{A geometric approach to the generalized Noether theorem} 

\author{Alessandro Bravetti}
\email{alessandro.bravetti@iimas.unam.mx} 
\affiliation{Instituto de Investigaciones en Matem\'aticas Aplicadas y en Sistemas, Universidad Nacional Aut\'onoma de M\'exico, A. P. 70543, M\'exico, DF 04510, Mexico}

\author{Angel Garcia-Chung}
\email{alechung@xanum.uam.mx} 
\affiliation{Departamento de F\'isica, Universidad Aut\'onoma Metropolitana - Iztapalapa, \\
San Rafael Atlixco 186, Ciudad de M\'exico 09340, M\'exico}
\affiliation{Universidad Panamericana, \\ 
Tecoyotitla 366. Col. Ex Hacienda Guadalupe Chimalistac, C.P. 01050 Ciudad de M\'exico, M\'exico}

\begin{abstract}
We provide a geometric extension of the generalized Noether theorem for scaling symmetries recently presented in~\cite{zhang2020generalized}.
Our version of the generalized Noether theorem has several positive features: it is  constructed in the most natural extension of the phase space, allowing for the symmetries to
be vector fields on such manifold and for the associated invariants to be first integrals of motion; it has a direct geometrical proof, paralleling the proof of the standard
phase space version of Noether's theorem; it automatically yields an inverse Noether theorem; it applies also to a large class of dissipative systems;
and finally, it allows for a much larger class of symmetries than just scaling transformations which 
form a Lie algebra, and are thus amenable to algebraic treatments.
\end{abstract}

\maketitle

\section{Motivation and previous  works}
\epigraph{
\emph{There are good reasons why the theorems should all be easy and the definitions hard. 
As the evolution of Stokes' Theorem revealed, a single simple principle can masquerade as several difficult results; 
the proofs of many theorems involve merely stripping away the disguise. 
The definitions, on the other hand, serve a twofold purpose: they are rigorous replacements for vague notions, and machinery for elegant proofs. \\
M. Spivak~\cite{spivak2018calculus}.}
}

Noether's theorem is one of the most profound and beautiful results in mathematical physics. 
This is because it is very easy to state and prove (in fact, it is one of the central topics  
in elementary courses on mechanics) but its consequences are far-reaching, ranging from the standard conservation
of energy and angular momentum in classical mechanics, to the existence of Noether charges in general relativity (one
of them being black holes' entropy~\cite{wald1993black}),
up to the reduction theorems in symplectic, Poisson, and contact geometry~\cite{abraham1978foundations}.

However, it is also well-known that there exist some important transformations that act like symmetries but 
are not Noether symmetries, for which it has long been believed that it does not exist an associated invariant quantity.
The most famous example of such symmetries is that of \emph{Kepler scalings}
	\begin{equation}\label{eq:KS}
	t\rightarrow \lambda^{3}t\,,\qquad {\bf q}\rightarrow \lambda^{2}{\bf q}\,,\qquad {S}\rightarrow \lambda{S}\,,\qquad \lambda=\text{const.}
	\end{equation}
Indeed, it is known that this is a {type of symmetry} of the Kepler problem (actually, this is the symmetry underlying Kepler's third law~\cite{zhang2019kepler}),
but it is not of Noether type, as it is clear because e.g.~the action and the dynamics are rescaled by the transformation. 
This type of symmetries that rescale the dynamics by multiplying it by a constant term
is sometimes referred to in the literature as a \emph{scaling symmetry} and they are examples of the more general \emph{dynamical similarities}~\cite{sloan2018dynamical}.

Let us remark that scaling symmetries have profound physical consequences. For instance, they have been employed to derive generalizations of the virial 
theorem~\cite{carinena2012geometric, carinena2013canonoid, zhang2020generalized}; moreover,
it has been recently argued that one can use such symmetries in order to
reduce a Hamiltonian system to a purely relational description in terms of the observables of the theory, where the (equivalent) dynamics can be shown to
be free of some spurious singularities such as the \emph{big bang singularity} of Einstein's equations~\cite{sloan2018dynamical,sloan2019scalar}.
Therefore it is of primary importance to have a general theory of such symmetries that may help recognize and classify them, and even more so
if such theory could help identify the associated conserved quantities (if any) that can then be used to perform the reduction of the system, as it is the case for
the standard Noether 
theorem.

Recently, in~\cite{zhang2020generalized} a generalized version of Noether's theorem that applies to scaling symmetries has been proved.
They proved that to any one-parameter family of transformations ${\bf q}\rightarrow {\bf q}'(t')$, $t\rightarrow t'$ that rescales the Lagrangian (and hence the action) as
	\begin{equation}\label{eq:rescaleZhang}
	L\left({\bf q}',\frac{d {\bf q}'}{d t'},t'\right)\frac{dt'}{dt}=\Lambda L\left({\bf q},\frac{d {\bf q}}{dt},t\right)
	\end{equation}
up to a boundary term, where $\Lambda$ is a constant, one can always associate a conserved quantity of the form
	\begin{equation}\label{eq:conservedZhang}
	Q={\bf p} \cdot \delta {\bf q}-H\delta t-S(t)\delta \Lambda\,,
	\end{equation}
where ${\bf p}=\frac{\partial L}{\partial \dot {\bf q}}$ are the conjugate momenta, $H$ is the Hamiltonian, and 
$S(t)=\int_{0}^{t}Ld\tau$ is the \emph{on-shell action}, namely, the action calculated along the trajectory.
For instance, in the case of Kepler scalings~\eqref{eq:KS}, the associated invariant~\eqref{eq:conservedZhang} is
	\begin{equation}\label{eq:conservedKepler}
	Q_{K}= 2\,{\bf p} \cdot {\bf q} - 3 t H_{K}-S(t)\,,
	\end{equation}
where 
	\begin{equation}
	H_K = \frac{{\bf p} \cdot {\bf p}}{2m} - \frac{4 \epsilon}{\sqrt{{\bf q} \cdot {\bf 	q}}}, \label{eq:HK}
	\end{equation}
is the Kepler Hamiltonian.

One of the most intriguing aspects of such generalized Noether theorem is the fact that the invariants~\eqref{eq:conservedZhang} 
are not functions on the phase space of the system, since they include also
the on-shell action variable. 
Therefore, a natural question arising about these invariants is the following
\begin{changemargin}{.5cm}{.5cm} 
	In what sense are the invariants~\eqref{eq:conservedZhang} first integrals of motion, that is, well-defined functions from the (possibly extended) phase space
	of the system to the real numbers that are constant along trajectories?
\end{changemargin} 
Or, more precisely,
\begin{changemargin}{.5cm}{.5cm} 
	Can we define the generalized Noether theorem geometrically in some appropriate (extension of the) phase space, so that
	the infinitesimal generators of scaling symmetries are actual vector fields on such manifold and their related invariants~\eqref{eq:conservedZhang} are actual
	functions on the manifold? 
\end{changemargin} 

Moreover, if the above is possible, it would be great if this geometric structure would allow for a simple proof of the generalized Noether theorem,
possibly paralleling the one of the standard Noether theorem, and also provide a direct geometric relationship between the infinitesimal generators and
their corresponding invariants.

{An interesting approach to solve some of the above puzzles  has been presented in~\cite{garcia2020vertical}, where the authors
construct a Noether theorem for scaling symmetries from the variational principle in the vertical extension of the standard phase space. 
The advantages of such approach are its general scope
and the fact that it effectively yields invariants associated with scaling symmetries which do not depend on computing the on-shell action. However,
the relationship between the original scaling symmetry and the constructed invariant is not so clear, since it turns out that the scaling symmetry 
in the original configuration space is not the Noether symmetry associated with the invariant, and that another symmetry can be found to be its associated
Noether symmetry. Moreover, the construction is based on performing first the variational formulation in the vertical extension of the tangent space and then projecting back 
down to the original manifold, which makes the derivation a bit obscure and can hamper the analysis of further questions such as e.g.~whether the new
symmetries obtained from this construction form a Lie algebra.
}

{On the other hand, the presence of the term $S(t)$ in~\eqref{eq:conservedZhang} is remindful of \emph{contact Hamiltonian mechanics}, 
where the contact phase space directly includes the on-shell action as an additional dynamical variable, thus serving as a natural arena in which the invariants (\ref{eq:conservedZhang}) 
can be described.}
{Indeed, a great deal of work has been devoted recently to the study of the so-called \emph{contact Hamiltonian systems}, which are 
Hamiltonian systems on contact manifolds, the ``odd-dimensional counterpart of symplectic manifolds'' 
(see e.g.~\cite{arnold2013mathematical, geiges2008introduction, wang2016implicit, bravetti2017contact, wang2019aubry, valcazar2018contact, bravetti2020invariant} 
for the formal theory
and~\cite{bravetti2018contact,Bravetti2017a,ciaglia2018contact, bravetti2020numerical} and the references therein for some applications). 
In this context, the standard symplectic phase 
space is extended to include an additional direction which in local coordinates turns out to be precisely the on-shell action of the system, and  there is a natural
way in which standard symplectic Hamiltonian systems can be embedded into contact Hamiltonian systems (and viceversa). Moreover contact Hamiltonian systems
can be derived from Herglotz' variational principle~\cite{georgieva2002first,georgieva2003generalized,vermeeren2019contact}, 
and a counterpart of the standard Noether theorem for such case has been proposed both in the variational
formulation~\cite{georgieva2002first,georgieva2003generalized,lazo2019noether} as well as in the geometric setting~\cite{de2020infinitesimal,gaset2020new}.
}

Motivated by all the above considerations, 
in this work we provide a geometric extension of the generalized Noether theorem for scaling symmetries recently presented in~\cite{zhang2020generalized}.
To do so, we use the geometric formulation for contact Hamiltonian systems in their
extended phase space (including time). 
In this way we will be able to prove very easily a generalized Noether theorem and its inverse that include as a particular case the association of scaling symmetries with their related invariants 
found in~\cite{zhang2020generalized}.
Moreover, we can prove directly that all the Noether symmetries defined in this way form a Lie algebra.
Further benefits of our approach are that the definition of generalized symmetries in this context
 allows for even more general symmetries than just scaling symmetries of the dynamics, and that it also
applies directly to a large class of dissipative systems.
Finally, in order to concretely show some of the consequences of our approach, we will present some paradigmatic examples in which the general theory 
recovers known results and reveals new interesting features.

\section{A brief introduction to symplectic and contact Hamiltonian mechanics}
In order to make this paper self-contained, we introduce here the very basic tools from symplectic and contact Hamiltonian systems that are needed in the following.
We refer to~\cite{bravetti2017contact,valcazar2018contact,bravetti2020invariant} 
for more detailed accounts and to~\cite{valcazar2018contact,gaset2020new} for the Lagrangian counterparts.

\subsection{The symplectic case}
We start with some {standard definitions, that can be found in any textbook on analytical mechanics, e.g.~\cite{abraham1978foundations}}. 

	\begin{definition}
The {\bf symplectic phase space} of a mechanical system is the cotangent bundle $T^{*}Q$, where $Q$ is the configuration manifold of the system, endowed with the canonical
symplectic form $\Omega=-d\alpha$, with $\alpha$ being the Liouville 1-form.
	\end{definition}
In local Darboux coordinates $(q^{a},p_{a})$, we have that $\alpha=p_{a}dq^{a}$, and thus $\Omega=dq^{a}\wedge dp_{a}$.
	\begin{definition}
	A {\bf symplectic Hamiltonian system} is a triple $(T^{*}Q,\Omega,H)$, with $H:T^{*}Q\rightarrow \mathbb R$ a sufficiently regular function called the symplectic Hamiltonian.
	\end{definition}
	\begin{definition}
 A {\bf symplectic Hamiltonian vector field $X_{H}$} is defined to be the only solution  to the condition
	\begin{equation}\label{eq:XH}
	\iota_{X_{H}}\Omega=-dH\,.
	\end{equation}
	\end{definition}
One can directly check that in Darboux coordinates this leads to the standard Hamilton equations 
	\begin{equation}\label{eq:Hameqs}
	\dot q^{a}=\frac{\partial H}{\partial p_{a}} \qquad \dot p_{a}=-\frac{\partial H}{\partial q^{a}}\,,
	\end{equation}
from which one can recover e.g.~the Newtonian dynamics of conservative systems.


Now we are ready to define Noether symmetries and conserved quantities.
	\begin{definition}\label{def:NSym}
	A {\bf {Noether} symmetry of a symplectic Hamiltonian system} is a vector field $Y\in \mathfrak{X}(T^{*}Q)$ such that
	$\iota_{Y}\Omega=-d F$ ($Y$ is Hamiltonian for some Hamiltonian function $F$) and $\mathfrak{L}_{Y}H=0$ ($Y$ {\bf preserves $H$}).
	\end{definition}
	\begin{definition}
	A {\bf conserved quantity} is a function $F:T^{*}Q\rightarrow\mathbb R$ such that $\mathfrak{L}_{X_{H}}F=0$.
	\end{definition}
	
Furthermore, we note that the condition~\eqref{eq:XH} provides an 
isomorphism between the Lie algebra of 
Hamiltonian vector fields on $T^{*}Q$ with the Lie bracket
and the Lie algebra of functions on $T^{*}Q$ with the Poisson bracket
	\begin{equation}\label{eq:Poisson}
	\{F,G\}_{P}:=\iota_{X_{F}}\iota_{X_{G}}\Omega\,.
	\end{equation}
Equipped with this isomorphism,
a Noether symmetry can equivalently be expressed as a Hamiltonian vector field $X_{F}$ such that $\{F,H\}_{P}=0$.
Thus,
 Noether's theorem can be proved in one line using the antisymmetry of the Poisson bracket, to obtain 
(see~\cite{butterfield2006symmetry} for more comments)
	\begin{theorem}[Symplectic Noether]\label{th:Noether1}
	$X_{F}$ is a {Noether} symmetry of a symplectic Hamiltonian system if and only if $F$ is a conserved quantity.
	\end{theorem}

Now, in order to understand Noether's theorem from a more general perspective, we need to introduce {some} general definitions.
	\begin{definition}\label{def:DynSymi}
	A {\bf {dynamical similarity of a vector field $X$}} is a vector field $Y\in \mathfrak{X}(T^{*}Q)$ such that $[Y,X]=\Lambda X$, for some (in general non-constant)
	function $\Lambda$.
	\end{definition}
As a particular but important case of the above, we have 
	\begin{definition}\label{def:DynSymm}
	A {\bf {dynamical symmetry of a vector field $X$}} is a vector field $Y\in \mathfrak{X}(T^{*}Q)$ such that $[Y,X]=0$.
	\end{definition}

It is easy to verify that any Noether symmetry is a dynamical symmetry of $X_{H}$; 
however the converse is not true, as the following proposition in the particular case $\lambda_{1}=\lambda_{2}$ clearly shows.
	\begin{proposition}\label{prop:ScalingSymm}
	Let $Y\in \mathfrak{X}(T^{*}Q)$ be such that $\mathfrak{L}_{Y}\Omega=\lambda_{2} \Omega$ ($Y$ is a {\bf non-strictly canonical symmetry}~\cite{carinena2012geometric,carinena2013canonoid,sloan2018dynamical})
	and $\mathfrak{L}_{Y}H=\lambda_{1} H$ ($Y$ rescales $H$).
	Then $[Y,X_{H}]=(\lambda_{1}-\lambda_{2})X_{H}$.
	\end{proposition} 
	\begin{proof}
	We have
	\begin{align}
	\iota_{[X_{H},Y]}\Omega=\iota_{X_{H}}\mathfrak{L}_{Y}\Omega-\mathfrak{L}_{Y}\iota_{X_{H}}\Omega=\lambda_{2}\iota_{X_{H}}\Omega+d\mathfrak{L}_{Y}H=(\lambda_{1}-\lambda_{2})dH\,,
	\end{align}
	and therefore, since $\Omega$ is non-degenerate, we conclude that $[Y,X_{H}]=(\lambda_{1}-\lambda_{2})X_{H}$.
	\end{proof}
Clearly, by the isomorphism described above, we know that we cannot associate any function of the phase space to a vector field that rescales $\Omega$ (these are not Hamiltonian),
and therefore we do not have much hope to extend Noether's theorem in order to include dynamical symmetries and similarities in the standard phase space.
As a particularly relevant example for our exposition, we first note that Kepler scalings~\eqref{eq:KS} can be shown to be induced by the action of the following vector field on the symplectic phase space 
	\begin{equation}\label{eq:KSgenerator1}
	Y_{KS}=2q^{a}\partial_{q^{a}}-p_{a}\partial_{p_{a}}\,. 
	\end{equation}
Then we have the following result
	\begin{proposition}\label{prop:KSsimilarity}
	$Y_{KS}$ is a dynamical similarity of Kepler's Hamiltonian vector field. 
	\end{proposition}
	\begin{proof}
	Let $H_{K}$ be the Kepler Hamiltonian~\eqref{eq:HK} and $\Omega_{K}$ the canonical symplectic form. 
	Then, as one can directly check, $\mathfrak{L}_{Y_{KS}}H_{K}=-2 H_{K}$ and
	$\mathfrak{L}_{Y_{KS}}\Omega_{K}= \Omega_{K}$, and therefore
	by Proposition~\ref{prop:ScalingSymm}, we have
	 $[Y_{KS},X_{H_{K}}]=-3X_{H_{K}}$.
	 \end{proof}
We remark that the result   
$[Y_{KS},X_{H_{K}}]=-3X_{H_{K}}$ appearing in the proof of Proposition~\ref{prop:KSsimilarity} is expected, 
as it means that the dynamics after the transformation induced by $Y_{KS}$ is rescaled by $\lambda^{-3}$, with $\lambda$
being the mock parameter along the flow of $Y_{KS}$. Therefore by
 scaling the time variable as in~\eqref{eq:KS}, we would observe no difference in the trajectories, and therefore we can consider $Y_{KS}$
 as a ``symmetry'' of the Kepler dynamics in some sense 
(note also that the proper scaling of the action in~\eqref{eq:KS} can be recovered by dimensional analysis).
However, 
as a consequence of Proposition~\ref{prop:KSsimilarity} and the above discussion, we conclude that it would be nonsensical to look for a conserved quantity associated to $Y_{KS}$
in the standard symplectic phase space.

\subsection{The contact case}
As a further generalization of the above results, we now provide analogue definitions for the contact case and proceed to prove the ``standard'' Noether theorem in this case.
	\begin{definition}
	The {\bf contact phase space} is the canonical contactification of the cotangent bundle, that is, the manifold $T^{*}Q\times\mathbb R$ endowed with the contact 1-form
$\eta=dS-\pi^{*}_{S}\alpha$, where $\pi_{S}:T^{*}Q\times\mathbb R\rightarrow T^{*}Q$ is the standard projection, and $S$ is the global coordinate on $\mathbb R$.
	\end{definition}
Associated to the contact 1-form $\eta$, there is a unique vector field, called the {\bf Reeb vector field}, defined by the conditions $\iota_{R}d\eta=0$ and $\iota_{R}\eta=1$.
In local Darboux coordinates $(q^{a},p_{a},S)$, we have that $\eta=dS-p_{a}dq^{a}$ and $R=\partial/\partial S$. 
	\begin{definition}
	A {\bf contact Hamiltonian system} is a triple $(T^{*}Q\times\mathbb R,\eta,h)$, with $h:T^{*}Q\times \mathbb R\rightarrow \mathbb R$ a sufficiently regular function called the contact Hamiltonian.
	\end{definition}
	\begin{definition}\label{def:CHVF}
 A {\bf contact Hamiltonian vector field $X_{h}$} is defined to be the only solution  to the conditions
	\begin{equation}\label{eq:Xh}
	\iota_{X_{h}}d\eta=dh-R(h)\eta\,\qquad \text{and} \qquad \iota_{X_{h}}\eta=-h\,.
	\end{equation}
	\end{definition}
\noindent One can directly check that in Darboux coordinates this leads to the contact Hamiltonian equations 
	\begin{equation}\label{eq:hameqs}
	\dot q^{a}=\frac{\partial h}{\partial p_{a}} \qquad 
	\dot p_{a}=-\frac{\partial h}{\partial q^{a}}-p_{a}\frac{\partial h}{\partial S} \qquad
	\dot S=p_{a}\frac{\partial h}{\partial p_{a}}-h\,,
	\end{equation}
from which one can recover the standard Hamilton equations~\eqref{eq:Hameqs} for $q^{a}$ and $p_{a}$ whenever $h$ does not depend on $S$.
Note also that from the last equation in~\eqref{eq:hameqs}, \emph{the new variable $S$ is the action of the system evaluated along the trajectories of the dynamics},
that is the \emph{on-shell action}. 
Moreover, the contact Hamiltonian equations~\eqref{eq:hameqs} can be used to model mechanical systems with different types of 
dissipative terms (see e.g.~\cite{bravetti2017contact,bravetti2016thermostat,gaset2020new}).

Now we are ready to define Noether symmetries and the analogue of conserved quantities in the contact case.
First we point out that the condition~\eqref{eq:Xh} provides a (global) 
isomorphism between the Lie algebra of 
contact vector fields on $T^{*}Q\times\mathbb R$ with the Lie bracket
and the Lie algebra of functions on $T^{*}Q\times\mathbb R$ with the Jacobi bracket
	\begin{equation}\label{eq:Jacobi}
	\{F,G\}_{J}:=\iota_{[X_{F},X_{G}]}\eta\,.
	\end{equation}
Then, starting from the symplectic analogue and from~\eqref{eq:Jacobi}, we have the following natural definitions. 
	\begin{definition}\label{def:Noethercontact1}
	A {\bf {Noether} symmetry of a contact Hamiltonian system} is a contact Hamiltonian vector field $X_{F}\in \mathfrak{X}(T^{*}Q\times\mathbb R)$ such that 
	$\{F,h\}_{J}=0$.
	\end{definition}

	
	

Moreover, we can generalize the concept of conserved quantities to the case of dissipative systems, which is the general case for contact systems, as follows.
	\begin{definition}\label{def:dissipq}
	A {\bf dissipated quantity of a contact Hamiltonian system} is a function $F:T^{*}Q\times\mathbb R\rightarrow\mathbb R$ that is dissipated at the same rate as the contact Hamiltonian, 
	 i.e.~$\mathfrak{L}_{X_{h}}F=-R(h)F$ (recall that $\mathfrak{L}_{X_{h}}h=-R(h)h$).
	 \end{definition}

Then, the contact version of Noether's theorem goes as follows~\cite{de2020infinitesimal,gaset2020new}: 
	\begin{theorem}[Contact Noether]\label{th:CNT}
	$X_{F}$ is a Noether symmetry of a contact Hamiltonian system if and only if $F=-\iota_{X_{F}}\eta$ is a dissipated quantity.
	\end{theorem}
	\begin{proof}
	The proof directly follows from the fact that
	
		$$
		\mathfrak{L}_{X_{h}}F=-\mathfrak{L}_{X_{h}}\left(\iota_{X_{F}}\eta\right)=-\iota_{X_{F}}\mathfrak{L}_{X_{h}}\eta-\iota_{\mathfrak{L}_{X_{h}}X_{F}}\eta=
		R(h)\iota_{X_{F}}\eta-\iota_{[X_{h},X_{F}]}\eta=-R(h)F-\iota_{[X_{h},X_{F}]}\eta\,.
		$$
	\end{proof}

Theorem~\ref{th:CNT} has been proved previously in this form in~\cite{de2020infinitesimal,gaset2020new}, {while similar results for the non-dissipative
case were already elucidated in~\cite{boyer2011completely}}.
Here we discuss a slight difference between our approach and theirs, mostly related to the definition of symmetries.
	In~\cite{gaset2020new} a more restrictive definition of symmetry has been considered, that is, $Y$ such that $[Y,X_{h}]=0$.
	This has been called an \emph{infinitesimal dynamical symmetry} in~\cite{gaset2020new} and Noether's theorem (therein called \emph{dissipation theorem})
	has been proved for such symmetries, analogously to the proof presented here.
	However, the condition for $Y$ to be a {Noether symmetry (Definition~\ref{def:Noethercontact1}) is more general (it allows for $[X_{F},X_{h}]\neq 0$),
	and therefore some symmetries may escape the more restrictive definition in~\cite{gaset2020new}.}
	On the other hand, in~\cite{de2020infinitesimal} the most general definition of symmetry has been considered that can lead to Theorem~\ref{th:CNT}. 
	Indeed, from the proof of 
	Theorem~\ref{th:CNT}, one finds out that the necessary and sufficient condition for the existence of a dissipated quantity associated to a vector field $Y$, is to have
	$\iota_{[X_{h},Y]}\eta=0$. This is indeed the definition of symmetry considered in~\cite{de2020infinitesimal}, where it has been referred to as a \emph{dynamical symmetry}.
	However, in such case the definition is too weak, and this hampers the 1:1 correspondence between dissipated quantities and symmetries which 
	constitutes part of the beauty of the standard Noether
	theorem. Note also that the term ``dynamical symmetry'' employed in such reference may be confusing {with respect to Definition~\ref{def:DynSymm}}. 
	Therefore we will refer to such symmetries as \emph{generalized dynamical symmetries} in the following discussion.
	For these reasons we prefer to consider a different concept of symmetry, Noether's symmetry, as stated in Definition~\ref{def:Noethercontact1}.
	Indeed, for a given dissipated quantity $F$, our definition selects the associated symmetry in the class of generalized dynamical symmetries corresponding to $F$ 
	that coincides exactly with the contact Hamiltonian vector field associated with $F$ (cf.~\cite[Remark~3]{de2020infinitesimal}).
	
	Moreover, the following important Corollary of Theorem~\ref{th:CNT} has been observed already both in~\cite{de2020infinitesimal} and in~\cite{gaset2020new}.	
	\begin{corollary}\label{cor:conserved1}
	Given two dissipated quantities $F_{1}$ and $F_{2}$, their ratio $F_{1}/F_{2}$, whenever it is well-defined, is a conserved quantity (note also that the contact
	Hamiltonian is, by definition, a dissipated quantity).
	\end{corollary}

At this point we are ready to focus on the comparison between Theorem~\ref{th:Noether1} and Theorem~\ref{th:CNT}, and on
discussing the role of dynamical similarities (including e.g.~Kepler scalings) in the contact case.

To begin, 	a direct comparison between the symplectic and the contact versions of Noether's theorem leads to the following observation: 
	 in the case where $h$ does not depend on $S$, both $h$ itself and any other dissipated quantity are in fact conserved quantities, and therefore we get
	 a generalization of the symplectic version
	of Noether's theorem, given that now both the symmetries and their associated invariants can depend explicitly on $S$, similarly to~\eqref{eq:conservedZhang}.
	However, 
	it is also clear that the dissipated (or conserved) quantities associated to Noether's symmetries in the contact case are functions of positions, momenta and the action
	only, and thus~\eqref{eq:conservedZhang} cannot be recovered, since in general they are functions of the time variable too. 
	To further illustrate this point, we proceed to show that, as in the symplectic case,
	Kepler scalings~\eqref{eq:KS} are not Noether symmetries according to any of the definitions discussed above.
	To do so, first we note that Kepler scalings are induced by the action of the following vector field on the contact phase space 
	\begin{equation}\label{eq:KSgenerator2}
	Y_{KS}^{c}=2q^{a}\partial_{q^{a}}-p_{a}\partial_{p_{a}}+S\partial_{S}\,. 
	\end{equation}
	Then we have the following result, whose proof is exactly analogous to that of the corresponding result in the symplectic case
	and can be also obtained by a direct calculation
	\begin{proposition}\label{prop:KSsimilarity2}
	$Y_{KS}^{c}$ is a non-trivial dynamical similarity of Kepler's contact Hamiltonian vector field. In particular $[Y_{KS}^{c},X_{h}^{K}]=-3X_{h}^{K}$,
	where $X_{h}^{K}$ is the contact Hamiltonian vector field associated to Kepler's Hamiltonian~\eqref{eq:HK}.
	\end{proposition}
Finally, we have the following important no-go result.
	\begin{proposition}\label{prop:dynsimnotsymm}
	Any non-trivial dynamical similarity of a contact Hamiltonian vector field is not a generalized dynamical symmetry. In particular they are not Noether symmetries.
	\end{proposition}
	\begin{proof}
	By definition of a non-trivial dynamical similarity, we have $[Y,X_{h}]=\Lambda X_{h}$, with $\Lambda\neq 0$. 
	Therefore $\iota_{[Y,X_{h}]}\eta=\Lambda\iota_{X_{h}}\eta=-\Lambda h\neq 0$.
	\end{proof}
	
We conclude that scaling symmetries cannot be Noether symmetries in the contact phase space, as it was the case in the symplectic phase space.
In the next section we show how to extend the contact version of Noether's theorem in order to include such symmetries.

\section{The generalized Noether theorem}
In this section we extend all the previous results to the case of generalized symmetries defined on the extended contact phase space and their associated dissipated (conserved) quantities.
We will see that this is at the same time an extension of the symplectic and contact versions of Noether's theorems, and of the generalized 
Noether theorem for scaling symmetries proved in~\cite{zhang2020generalized}.

As usual, we start with the necessary definitions.
In order to consider contact Hamiltonian systems with an explicit time dependence, {we follow~\cite{bravetti2017contact} and define the following.} 
	\begin{definition}
	We call the {\bf extended contact phase space}
	the manifold $T^{*}Q\times\mathbb R\times\mathbb R$, endowed with the 1-form $\eta^{E}=\pi^{*}_{t}\eta+hdt$, where 
	$\pi_{t}:T^{*}Q\times\mathbb R\times\mathbb R\rightarrow T^{*}Q\times\mathbb R$ is
	the projection and $t$ is the coordinate on $\mathbb R$.
	\end{definition}
Note that this extension is not canonical, in the sense that it depends on the system at hand through $h$, and that $(T^{*}Q\times\mathbb R\times\mathbb R,\eta^{E})$ in general 
is a \emph{pre-symplectic manifold}~\cite{marmo2020symmetries}. Indeed, $d\eta^{E}$ is non-degenerate (hence symplectic) if and only if {$\partial h/\partial S\neq 0$}.

Moreover, we define time-dependent contact Hamiltonian vector fields as follows.
	\begin{definition}
{A \bf time-dependent contact Hamiltonian vector field $X_{h}^{t}$} is the solution 
to the conditions 
	\begin{equation}\label{eq:XhE}
	\iota_{X_{h}^{t}}d\eta^{E}=-R(h)\eta^{E}\, \quad \text{and} \quad \iota_{X_{h}^{t}}\eta^{E}=0\,,
	\end{equation}
	\end{definition}
\noindent {where $R(h) = \frac{\partial h}{\partial S}$.} 
We point out that this definition is different from that of a pre-symplectic system, as the right hand side of the first condition in~\eqref{eq:XhE} is not a closed 1-form~\cite{marmo2020symmetries}.
Note that in this case ${X_{h}^{t}}$ is not uniquely fixed by the two conditions in~\eqref{eq:XhE}, as one has the freedom to choose a global function 
$f:T^{*}Q\times \mathbb R\times \mathbb R\rightarrow \mathbb R-\{0\}$ 
that multiplies ${X_{h}^{t}}$,
that is, \emph{${X_{h}^{t}}$ is uniquely defined up to (nowhere-vanishing) rescalings.}
For instance, fixing $f=1$, in local adapted coordinates one recovers the contact Hamiltonian equations for time-dependent Hamiltonians
	\begin{equation}\label{eq:hameqswitht}
	\dot q^{a}=\frac{\partial h}{\partial p_{a}} \qquad 
	\dot p_{a}=-\frac{\partial h}{\partial q^{a}}-p_{a}\frac{\partial h}{\partial S} \qquad
	\dot S=p_{a}\frac{\partial h}{\partial p_{a}}-h\,,\qquad
	\dot t=1\,,
	\end{equation}
from which we infer that changing the function $f$ amounts to reparametrizing the dynamics.

To avoid clutter of notation and without loss of generality, 
we will also assume from now on that $f=1$,
that is, $\X=X_{h}+\partial_{t},$ where $X_{h}$ is the contact Hamiltonian vector field associated with $h$. 
All the results can be easily generalized by including the 
factor $f$ back into the corresponding equations.

Now we are ready to study symmetries and associated dissipated (conserved) quantities in this extended phase space.
	\begin{definition}\label{def:symm2}
	We call $Y\in\mathfrak{X}(T^{*}Q\times\mathbb R\times\mathbb R)$ a {\bf generalized Noether symmetry of a time-dependent contact Hamiltonian system} if 
		$\L_{Y}\eta^{E}=\lambda\,\eta^{E}$ for some $\lambda\in C^{\infty}(T^{*}Q\times\mathbb R\times\mathbb R)$.
	\end{definition}

\noindent One can check that this definition includes the previous definitions of Noether symmetries of symplectic and contact Hamiltonian systems
(cf.~Theorem~\ref{th:IGNT2} below). 
However, in this case we allow also
for the symmetry to have a non-zero time component, and therefore such symmetries naturally allow for time rescalings. 

We have the following important property.
	\begin{proposition}\label{prop:liealg}
	Generalized Noether symmetries form a Lie algebra with the Lie bracket.
	\end{proposition}
	\begin{proof}
	{Let $\L_{Y_{1}}\eta^{E}=\lambda_{1}\eta^{E}$ and $\L_{Y_{2}}\eta^{E}=\lambda_{2}\eta^{E}$. Then}
	$$\L_{[Y_{1},Y_{2}]}\eta^{E}=\L_{Y_{1}}\L_{Y_{2}}\eta^{E}-\L_{Y_{2}}\L_{Y_{1}}\eta^{E}=\left(Y_{1}(\lambda_{2})-Y_{2}(\lambda_{1})\right)\eta^{E}\,.$$
	\end{proof}

We can also define dissipated quantities in analogy with the standard contact case.
	\begin{definition}\label{def:dissipatedt}
	A {\bf dissipated quantity in the extended contact phase space} 
	is a function $F:T^{*}Q\times\mathbb R\times\mathbb R\rightarrow\mathbb R$ that satisfies
	 \begin{equation}\label{eq:Killing}
	 \mathfrak{L}_{X_{h}^{t}}F=-R(h)F\,.
	 \end{equation}
	 In the following we will refer to condition~\eqref{eq:Killing}  as {\bf the dissipation equation}.
	 \end{definition}
	
Note that the Hamiltonian is not necessarily a dissipated quantity itself. Indeed, as it happens in the symplectic case where $H$ is not conserved whenever it depends explicitly
on $t$, here $h$ is not a dissipated quantity whenever it depends explicitly on $t$. 
	 

Next, we have the following Lemma.
	\begin{lemma}\label{prop:symmsubset}
	Let Y be a generalized Noether symmetry. Then $\iota_{[Y,\X]}\eta^{E}=0$. 
	\end{lemma}
	\begin{proof}
	$$\iota_{[Y,\X]}\eta^{E}=[\L_{Y},\iota_{\X}]\eta^{E}=\L_{Y}\iota_{\X}\eta^{E}-\iota_{\X}\L_{Y}\eta^{E}=-\iota_{\X}\lambda\eta^{E}=0\,.$$
	\end{proof}

Now we are ready to prove the generalized Noether Theorem.
	\begin{theorem}[Generalized Noether Theorem]\label{th:GNT2}
		Let $Y$ be a generalized Noether symmetry of a time-dependent contact Hamiltonian system. Then $F=-\iota_{Y}\eta^{E}$ is a dissipated quantity.
	\end{theorem}
As usual, the proof is a one-liner:
	\begin{proof}
	$$
	\mathfrak{L}_{X_{h}^{t}}F=-\mathfrak{L}_{X_{h}^{t}}\left(\iota_{Y}\eta^{E}\right)=-\iota_{Y}\mathfrak{L}_{X_{h}^{t}}\eta^{E}-\iota_{\mathfrak{L}_{X_{h}^{t}}Y}\eta^{E}=
	R(h)\iota_{Y}\eta^{E}-\iota_{[\X,Y]}\eta^{E}=-R(h)F-\iota_{[\X,Y]}\eta^{E}\,.
	$$
	\end{proof}
We remark that the dissipated quantity associated with a generalized Noether symmetry is obtained exactly in the same way as in Theorem~\ref{th:CNT}, with 
$\eta$ replaced by $\eta^{E}$. This replacement is fundamental in order to obtain dissipated (conserved) quantities associated to scaling symmetries of the type obtained
in~\cite{zhang2020generalized} (cf.~Equation~\eqref{eq:conservedZhang}).

Clearly, we could have proved Theorem~\ref{th:GNT2} even under the more general definition of symmetry  $\iota_{[Y,\X]}\eta^{E}=0$.
However, in such case we would lose the possibility to associate to any invariant a somewhat ``unique'' symmetry, as we now show.
	\begin{theorem}[Inverse generalized Noether Theorem]\label{th:IGNT2}
	Let $F$ be a dissipated quantity. Then
	\begin{equation}\label{eq:YF}
	Y_{F}=X_{F}+Y^{t}\X
	\end{equation} 
	is the most general form of its associated generalized Noether symmetry, 
	where $X_{F}$ is the contact Hamiltonian vector field associated to $F$ (see Definition~\ref{def:CHVF}) and 
	$Y^{t}$ is a free function.
	\end{theorem}
		\begin{proof}
	Given $F$, we wish to find the corresponding $Y$ that satisfies
		$\L_{Y}\eta^{E}=\lambda\eta^{E}$ ($Y$ generalized Noether symmetry) and 
		$\iota_{Y}\eta^{E}=-F\,,$ (the associated dissipated quantity is $F$).
	Using these two conditions and Cartan's identity, we get immediately $\iota_{Y}d\eta^{E}=d F+\lambda\eta^{E}$, which, when written in Darboux coordinates 
	is equivalent to the following conditions
		\begin{equation}\label{eq:lambda}
		\lambda =-Y^{t}\frac{\partial h}{\partial S}-\frac{\partial F}{\partial S}
		\end{equation}
	and
		\begin{align}
		Y^{a}&=Y^{t}\frac{\partial h}{\partial p_a}+\frac{\partial F}{\partial p_a}\label{eq:Yq}\\
		Y_{a}&=-Y^{t}\left[\frac{\partial h}{\partial q^a}+p_a \frac{\partial h}{\partial S}\right]-\frac{\partial F}{\partial q^a}-p_a \frac{\partial F}{\partial S}\label{eq:Yp}\\
		Y^{a}\frac{\partial h}{\partial q^a}+Y_{a}\frac{\partial h}{\partial p_a}+Y^{S}\frac{\partial h}{\partial S}&=-\left(Y^{t}\frac{\partial h}{\partial S}+\frac{\partial F}{\partial S}\right) {h} +\frac{\partial F}{\partial t}\,.\label{eq:Fdissipated}
		\end{align}
	Now from $\iota_{Y}\eta^{E}=-F$ and the above conditions, we get the additional requirement
		\begin{equation}\label{eq:Ys}
		Y^{S}=Y^{t}\left(p_a \frac{\partial h}{\partial p_a}-h\right)+p_a \frac{\partial F}{\partial p_a}-F\,.
		\end{equation}
	From~\eqref{eq:Yq}, \eqref{eq:Yp} and \eqref{eq:Ys} we get that $Y=X_{F}+Y^{t}\left(X_{h}+\partial_{t}\right)=X_{F}+Y^{t}\X$, while equation~\eqref{eq:Fdissipated} can be rewritten after some algebra
	as $\X(F)=-\frac{\partial h}{\partial S}F$, that is, the dissipation equation~\eqref{eq:Killing}. 
	\end{proof}

Theorems~\ref{th:GNT2} and~\ref{th:IGNT2} are the main contributions of this work. Let us make some comments about these results:
first,	 
	we get a generalization of the contact version of Noether's theorem, as now the 
	symmetries and the associated invariants can depend on $t$ explicitly. 
	Additionally, if $h$ does not depend on $S$ and $t$, then we get a generalization of the symplectic 
	version of Noether's theorem, given that the symmetries and the associated invariants now can depend on $S$ and $t$ explicitly.
	This will be the case of e.g.~Kepler scalings.
Secondly,	
	as anticipated, in the most general case $\mathfrak{L}_{X_{h}^{t}}h=-R(h)h+\frac{\partial h}{\partial t}$, and therefore whenever $h$ depends explicitly on $t$,
	 it is not a dissipated quantity itself. This implies that whenever $h$ depends on $t$ we lose the fact that $F/h$ is a conserved quantity, as it happened in the contact
	 version of Noether theorem.
	However, whenever $h$ does not depend on $t$ we have that $F/h$ is still a conserved quantity, even when $F$ depends
	 on $t$. Moreover, we still have the following analogue of Corollary~\ref{cor:conserved1}.
	 \begin{corollary}\label{cor:conserved2}
	 Given two dissipated quantities $F$ and $G$, the function $F/G$, whenever defined, is a conserved quantity.
	 \end{corollary}

Finally,
	note that for $Y_{F}$ as in~\eqref{eq:YF}, we have $\iota_{Y_{F}}\eta^{E}=-F$. 
	Moreover observe that, since $\iota_{Y^{t}\X}\eta^{E}=0$ and $\L_{Y^{t}\X}\eta^{E}=-Y^{t}R(h)\eta^{E}$,  
	we can always add a term of the form $Y^{t}\X$
	to any given generalized Noether symmetry and obtain another generalized Noether symmetry corresponding to the same invariant.
	Therefore such ``gauge freedom'' in the fixing of the symmetry is unavoidable.
	We say that \emph{the correspondence between symmetries and invariants is 1:1 up to the mentioned gauge freedom}.



A fundamental question for our discussion at this point is the 
relationship between generalized Noether symmetries and the dynamical similarities of $\X$. 
This is the content of the next two propositions.
We start with a lemma.
	\begin{lemma}\label{le:kernels}
	Let $A\in \mathfrak{X}(T^{*}Q\times\mathbb R\times\mathbb R)$ be such that 
		\begin{itemize}
		\item[i)] $A\in\ker\eta^{E}$ 
		\item[ii)] $\iota_{A}d\eta^{E}=f_{A}\eta^{E}$\,, for some function $f_{A}$.
		\end{itemize}
		Then
		$A=A^{t}\X$, with $f_{A} = -(\frac{\partial h}{\partial S}) A^{t}$. 
	\end{lemma}	
	\begin{proof}
	We have
	$A\in\ker\eta^{E}\implies A^{S}=p_a A^{a}-hA^{t}$\,. 
	Then we can directly compute that 
	\begin{align}
	\iota_{A}d\eta^{E}=f_{A}\eta^{E} \implies A^{a}&=A^{t}\frac{\partial h}{\partial p_a}\\
							A_{a}&=-A^{t}\frac{\partial h}{\partial q^a}+f_{A}p_a \\
							f_{A}&=-\frac{\partial h}{\partial S}A^{t}\label{eq:At}\\
							f_{A}h&=A^{a}\frac{\partial h}{\partial q^a}+A_{a}\frac{\partial h}{\partial p_a}+A^{S}\frac{\partial h}{\partial S}\,.
	\end{align}
	Using \eqref{eq:At} we can rewrite the other equations as
	\begin{align}
	A^{a}&=A^{t}\frac{\partial h}{\partial p_a}\\
	A_{a}&=-A^{t}\left(\frac{\partial h}{\partial q^a}+p_a \frac{\partial h}{\partial S}\right)\\
	A^{S}&=A^{t}\left(p_a \frac{\partial h}{\partial p_a}-h\right)\,,
	\end{align}
	meaning that 
	$A=A^{t}X_{h}+A^{t}\partial_{t}=A^{t}\X$, as claimed.
	\end{proof}
Now we can prove the following important implication, which states that generalized Noether symmetries are a subset of the dynamical similarities of $\X$.
	\begin{proposition}\label{prop:symmVSrescalings}
	$Y$ generalized Noether symmetry 
	$\implies$ $[Y,\X]=\tilde \lambda \X$.
	\end{proposition}
	\begin{proof}
	To prove the result, we will prove that $[Y,\X]\in\ker\eta^{E}$ and that $\iota_{[Y,\X]}d\eta^{E}=f\,\eta^{E}$ for some function $f$,
	and then we use Lemma~\ref{le:kernels}.
	We start with
	\begin{equation}\label{eq:part1}
	\L_{Y}\eta^{E}=\lambda\eta^{E}\implies \iota_{Y}d\eta^{E}-dF=\lambda\eta^{E}
	\end{equation}
	On the other hand, 
	\begin{align}
	\iota_{[\X,Y]}\eta^{E}&=\L_{\X}\iota_{Y}\eta^{E}-\L_{Y}\cancel{\iota_{\X}\eta^{E}}-\iota_{\X}\iota_{Y}d\eta^{E}
					=-\iota_{\X}dF+\iota_{Y}\iota_{\X}d\eta^{E}\underset{\eqref{eq:part1}}{=}\lambda\iota_{\X}\eta^{E}=0\,,
	\end{align}
	and therefore $[\X,Y]\in\text{ker}\,\eta^{E}$.
	Finally, a direct computation shows that 
	\begin{align}
	\iota_{[Y,\X]}d\eta^{E}=\L_{[Y,\X]}\eta^{E}=\L_{Y}\L_{\X}\eta^{E}-\L_{\X}\L_{Y}\eta^{E}=-\left[\L_{Y}R(h)+\L_{\X}\lambda\right]\eta^{E}
	\end{align}
	and therefore by Lemma~\ref{le:kernels} we conclude.
	\end{proof}
	
	{The following example illustrates that in general the inclusion defined by Proposition~\ref{prop:symmVSrescalings} is strict,
	i.e., generalized Noether symmetries are a proper subset of dynamical similarities.
	\begin{example}
	Let $X_{K}^{t}$ be the contact Hamiltonian vector field associated with the Kepler Hamiltonian~\eqref{eq:HK} in the extended contact phase space
	and consider the vector field 
	$Y=H_{K}\partial_{S}$. One can directly check that $[Y,X_{K}^{t}]=0$, and thus $Y$ is a (trivial) dynamical similarity. However, it is not a generalized 
	Noether symmetry, as $\L_{Y}\eta^{E}\neq \lambda \eta^{E}$ for any function $\lambda$.
	\end{example}
	}
	
	{Despite the above remark,
	we have the next proposition, which is an immediate corollary of Theorem~\ref{th:IGNT2} and provides} in some sense the inverse of Proposition~\ref{prop:symmVSrescalings},
	as it guarantees that for any dynamical similarity there {exists} an associated dissipated (conserved)
	quantity, and thus a related generalized Noether symmetry.
	\begin{proposition}
	Let $Y$ be such that $[Y,\X]= \lambda \X$. Then  $F:=-\iota_{Y}\eta^{E}$ is a dissipated quantity and there exists a generalized Noether symmetry associated with $F$.
	\end{proposition}
	\begin{proof}
	Since $[Y,\X]= \lambda \X$, then $\iota_{[Y,\X]}\eta^{E}=0$ and thus from the proof of Theorem~\ref{th:GNT2} we know that $F$ is a dissipated quantity.
	Hence, from Theorem~\ref{th:IGNT2} it follows that there is a generalized Noether symmetry associated to $F$. 
	\end{proof}
\noindent At this point we have reached our aim, that is, we have found the appropriate extension of the phase space
and the appropriate definition of (generalized) Noether symmetries in such space such that
we can guarantee that to any dynamical similarity of $\X$ we can associate a dissipated (conserved) quantity via the generalized Noether theorem. 
This will be the 
case of e.g.~Kepler scalings, as we will prove shortly.
	
\section{Scalings as generalized Noether symmetries}

To keep the discussion as general as possible, let us consider a contact Hamiltonian of the form
	\begin{equation}
	h = \frac{1}{2m} \,{\bf p} \cdot {\bf p} + f(t) V({\bf q}) + g(S)\,, \label{GFHamil}
	\end{equation}	
\noindent where 
	\begin{equation}
	V(\xi \, {\bf q}) = \xi^k V({\bf q}), \qquad g(\xi S) = \xi^\kappa g(S), 
	\end{equation}
\noindent i.e., $V( {\bf q})$ and $g(S)$ are homogeneous functions of degree $k$ and $\kappa$ respectively. 
The functions $f(t)$ and $g(S)$ are so far arbitrary functions of time $t$ and the variable $S$. 

Consider now a general scaling transformation
\begin{equation}
{\bf q}' = \zeta^\alpha {\bf q}, \qquad {\bf p}' = \zeta^\beta {\bf p}, \qquad S' = \zeta^\gamma S, \qquad t' = \zeta^\sigma t,  \label{DefScal}
\end{equation}
\noindent where $\zeta \in \mathbb{R}_{> 0}$ is the scaling parameter. The corresponding infinitesimal transformation is of the form 
\begin{equation}
\delta {\bf q} = \left.\left(\frac{\partial }{\partial \zeta} {\bf q}' \right)\right\vert_{\zeta =1} = \alpha {\bf q}, \quad \delta {\bf p} = \left.\left(\frac{\partial }{\partial \zeta} {\bf p}' \right)\right\vert_{\zeta =1} = \beta {\bf p}, \quad \delta S = \left.\left(\frac{\partial }{\partial \zeta} S' \right)\right\vert_{\zeta =1} = \gamma S, \quad \delta t = \left.\left(\frac{\partial }{\partial \zeta} t' \right)\right\vert_{\zeta =1} = \sigma t, 
\end{equation}
\noindent and thus the associated generator 
	$$Y=Y^{a}\partial_{q^{a}}+Y_{a}\partial_{p_{a}}+Y^{S}\partial_{S}+Y^{t}\partial_{t}$$
 has components
\begin{equation}
Y^a = \alpha q^a, \qquad Y_a = \beta p_a , \qquad Y^{S} = \gamma S, \qquad Y^t = \sigma t. \label{GScalingVector}
\end{equation}
\noindent We now insert these components in \eqref{eq:YF} and obtain the following equations
\begin{eqnarray}
\alpha q^a &=&  \frac{\partial F}{\partial p_a} + Y^{t} \frac{p_a}{m} , \label{QEq}\\
\beta p_a &=& - \frac{\partial F}{\partial q^a} - p_a \frac{\partial F}{\partial S} - Y^{t} \left( f(t) \frac{\partial V}{\partial q^a} +p_a \frac{d g}{dS}  \right), \label{PEq} \\
\gamma S &=&  p_a \frac{\partial F}{\partial p_a} - F + Y^{t} \left(  \frac{1}{2m}  {\bf p} \cdot {\bf p} - f(t) V(\vec{q}) -g(S) \right) , \label{SEq} \\ 
\sigma t &=& Y^{t}. \label{TimeEq}
\end{eqnarray}

In order to find the solution, we proceed as follows. 
First, we replace equation \eqref{TimeEq} into (\ref{QEq}), (\ref{PEq}) and (\ref{SEq}). This yields the new set of equations
\begin{eqnarray}
\alpha q^a &=&  \frac{\partial F}{\partial p_a} + \sigma t \frac{p_a}{m} , \label{QEq2}\\
\beta p_a &=& - \frac{\partial F}{\partial q^a} - p_a \frac{\partial F}{\partial S} - \sigma t \left( f(t) \frac{\partial V}{\partial q^a} +p_a \frac{d g}{dS}  \right), \label{PEq2} \\
\gamma S &=&  p_a \frac{\partial F}{\partial p_a} - F + \sigma t \left(  \frac{1}{2m}  {\bf p} \cdot {\bf p} - f(t) V(\vec{q}) -g(S) \right). \label{SEq2}
\end{eqnarray}
Then, we solve for $\frac{\partial F}{\partial p_a}$ in equation (\ref{QEq2}) and obtain
\begin{equation}
\frac{\partial F}{\partial p_a} = \alpha q^a - \sigma t \frac{p_a}{m}   , 
\end{equation}
\noindent which can now be inserted in (\ref{SEq2}) to obtain the associated dissipated quantity
\begin{equation}
F  = \alpha\, {\bf q} \cdot {\bf p} -  \sigma t\, h({\bf q}, {\bf p}, S, t)  - \gamma S\,, \label{ExpforF}
\end{equation}
which should now be compared to~\eqref{eq:conservedZhang}.
The third step is to insert this result into equation (\ref{PEq2}), which leads to the following condition
\begin{equation}
\left( \beta - \gamma + \alpha\right)  + \sigma t  \frac{dg}{dS} =0. \label{betaCond}
\end{equation}
\noindent As $g(S)$ is a time independent function, there are only two possibilities to find a solution of~\eqref{betaCond}: 
(i) $\frac{dg}{dS} = 0$, which implies $g(S) = g_0=\mbox{constant}$, or (ii) $\sigma = 0$. Let us analyze each case separately.

\subsection{Case (i): Homogeneous non-dissipative time-dependent systems}

In this case $g(S) =\mbox{constant}$ and this constant can be taken to be zero without loss of generality. 
Therefore, the Hamiltonian function (\ref{GFHamil}) takes the following expression
\begin{equation}
h = \frac{1}{2m}\, {\bf p} \cdot {\bf p} + f(t) V({\bf q}). \label{GFHamil2}
\end{equation}	
Recall that since in this case $h$ does not depend on $S$, then any dissipated quantity is in fact a conserved one.

We now return to (\ref{betaCond}) and notice that it can be used to fix $\beta$ in terms of $\alpha$ and $\gamma$
\begin{equation}
 \beta = \gamma - \alpha. \label{beta}
\end{equation}
\noindent Now, for $F$ as in~\eqref{ExpforF} and $h$ as above, we can rewrite the dissipation equation~\eqref{eq:Killing} as
\begin{equation}
\left( \alpha - \frac{\gamma}{2} - \frac{\sigma}{2} \right) \frac{{\bf p} \cdot {\bf p}}{m}  + (\gamma - \sigma)f(t) V({\bf q}) - \alpha f(t) {\bf q} \cdot \nabla V - \sigma t \dot{f}(t) V= 0. \label{KillingEqCond}
\end{equation}
\noindent Clearly, the coefficient of the ${\bf p} \cdot {\bf p}$ term must be zero, hence
\begin{equation}
\alpha = \frac{\gamma}{2} + \frac{\sigma}{2}.
\end{equation}
\noindent On the other hand, the homogeneity of the potential $V({\bf q})$ yields the following relation
\begin{equation}
{\bf q} \cdot \nabla V({\bf q}) = k V({\bf q}), \label{HomoCond}
\end{equation}
\noindent which is used to replace ${\bf q} \cdot \nabla V({\bf q}) $  in (\ref{KillingEqCond}). As a result we obtain the condition
\begin{equation}
f(t) \left[  \gamma \left( 1 - \frac{k}{2} \right) - \sigma \left( 1+ \frac{k}{2}  \right) \right] = \sigma t \dot{f}(t) . \label{KillingEqCond2}
\end{equation}
We can now consider three cases: 

{\bf Case (1).} If $\sigma = 0$, then (\ref{KillingEqCond2}) gives
\begin{equation}
  \gamma \left( 1 - \frac{k}{2} \right) = 0 . \label{KEqCond1}
\end{equation}
\noindent which admits two solutions: $\gamma = 0$ and $k=2$. 
The first option yields a trivial solution, $F=0$, and the second solution gives a {conserved quantity} $F$ of the form
\begin{equation}
F_0({\bf q}, {\bf p}, S)  =  {\bf q} \cdot {\bf p} - 2S. \label{F0}
\end{equation}

{\bf Case (2).} If $\dot{f}(t) = 0$, then (\ref{KillingEqCond2}) results in
\begin{equation}
  \gamma \left( 1 - \frac{k}{2} \right) - \sigma \left( 1+ \frac{k}{2}  \right) = 0, \label{KEqCond2}
\end{equation}
\noindent which admits three possible solutions:
\begin{enumerate}
\item[2.1)] If $k = 2$, then $\sigma =0$, and this recovers the second solution of {\bf Case (1)}, where $F_0$ was defined.
\item[2.2)] If $k=-2$, then $\gamma = 0$, and the conserved quantity takes the form
\begin{equation}
F_1({\bf q}, {\bf p}, t) = {\bf q} \cdot {\bf p} - 2 t\, h({\bf q}, {\bf p},t).
\end{equation}
Note that in this particular case the dependence on $S$ in the conserved quantity disappears. 
This is because in this case scaling transformations 
are canonical symmetries~\cite{jackman2016no}.
\item[2.3)] If $k \neq \pm 2$, then $\gamma = \frac{(2+k)}{(2-k)} \sigma$. In this case, the conserved quantity is
\begin{equation}
F_2({\bf q}, {\bf p}, S, t) = \frac{2}{2-k} {\bf q} \cdot {\bf p} - t\, h({\bf q}, {\bf p},t) - \frac{(2+k)}{(2-k)} S\,.
\end{equation}
\end{enumerate}
Note that this case contains the Kepler system with Hamiltonian~\eqref{eq:HK}. 
Given that in such case the potential is homogeneous of degree $k=-1$, the conserved quantity takes the form
	\begin{equation}
	F^{(K)}_2 = \frac{2}{3} {\bf q} \cdot {\bf p} - t H_K - \frac{1}{3} S\,,
	\end{equation}
which coincides with~\eqref{eq:conservedKepler} up to multiplication by an irrelevant factor of $3$. 
{Observe also that using~\eqref{eq:conservedKepler} in~\eqref{eq:YF} and choosing $Y^{t}=3t$ one obtains
$$Y_{Q_{K}}=3t\,\partial_{t}+2q^{a}\partial_{q^{a}}-p_{a}\partial_{p_{a}}+S\partial_{S}\,,$$
which is precisely the generator of Kepler scalings in the extended contact phase space, thus confirming that these are generalized Noether symmetries.
}

{\bf Case (3).} 
If $\sigma, \dot{f}(t) \neq 0$, then (\ref{KillingEqCond2}) admits a solution for a very specific function $f(t)$, depending on the parameter $\Lambda : = \gamma / \sigma$ as follows
\begin{equation}
f\left( t; \Lambda \right) = t^{ \frac{ (2-k) \Lambda }{2} - \frac{(2+k)}{2}}\,, \label{KEqCond3}
\end{equation}
\noindent and $F$ is of the form
\begin{equation}
F({\bf q}, {\bf p}, S, t) = \Lambda F_0({\bf q}, {\bf p}, S) +  F_1({\bf q}, {\bf p}, t; \Lambda)\,. \label{Case3Inv}
\end{equation}
As a particular example of this case, we can consider a time-dependent Kepler problem with the following Hamiltonian function 
\begin{equation}
H_{t K}(t) = \frac{{\bf p} \cdot {\bf p}}{2m} - f^{(K)}(t;\Lambda)\frac{4 \epsilon}{\sqrt{{\bf q} \cdot {\bf q}}}\,, \label{TDKeplerJamil}
\end{equation}
\noindent where $f^{(K)}(t;\Lambda)$ can be derived using (\ref{KEqCond3}) and recalling that $k=-1$, that is,
\begin{equation}
f^{(K)}\left( t; \Lambda \right) = t^{ \frac{ 3 \Lambda - 1}{2} }\,. \label{KEqCond4}
\end{equation}
\noindent From (\ref{Case3Inv}) we obtain the following expression for the conserved quantity for this system
\begin{equation}
F({\bf q}, {\bf p}, S, t) =  ( \Lambda + 1 ) {\bf q} \cdot {\bf p} -2H_{tK}t- 2 \Lambda S\,. \label{TDKeplerInv}
\end{equation}

\subsection{Homogeneous dissipative time-dependent systems}	

If $\frac{dg}{d S} \neq 0$, then $\sigma =0$ in (\ref{betaCond}) and this yields a condition for $\beta$ as given in \eqref{beta}. 
Moreover, the Hamiltonian in this case is still of the form (\ref{GFHamil}) but the expression for the function $F$ is now given as
\begin{equation}
F  = \alpha {\bf q} \cdot {\bf p}   - \gamma S. \label{ExpforF2}
\end{equation}
Inserting this expression into the dissipation equation~\eqref{eq:Killing}, we obtain
\begin{equation}
\left( \alpha - \frac{\gamma}{2}\right) \frac{{\bf p} \cdot {\bf p}}{m} - \alpha f(t) k V({\bf q}) + \gamma f(t) V({\bf q}) + \gamma g(S) - \gamma S \frac{d g}{dS} = 0\,, \label{NewKillingEq}
\end{equation}
\noindent where the homogeneity condition of the potential (\ref{HomoCond}) was used. 
As before, the coefficient of ${\bf p} \cdot {\bf p}$ has to be zero, which means 
\begin{equation}
\alpha  = \frac{\gamma}{2}\,,
\end{equation}
\noindent and now (\ref{NewKillingEq}) takes the form
\begin{equation}
\gamma f(t) k V({\bf q}) \left( 1 - \frac{k}{2} \right) + \gamma g(S) \left( 1- \kappa \right) = 0\,,
\end{equation}
\noindent where $\kappa$ is the homogeneity degree of $g(S)$.

Again, there are three cases to be considered to obtain solutions of this algebraic equation. 
The case in which $\gamma = 0$ gives the trivial solution, $F = 0$. The other two cases, (i) $k = 0$ and $\kappa =1$ and (ii) $k = 2$ and $\kappa =1$ yield the same dissipated function $F$ which turns out to be $F_0$ (cf.~equation~\eqref{F0}). Remarkably, in both cases the dissipative term $g(S)$ is forced to be of the form $g(S) = g_0 S$. 

In summary, a Hamiltonian of the type \eqref{GFHamil} admits a scaling symmetry \eqref{DefScal}
only in cases where $g(S)$ is a linear function and the homogeneity degree of the potential is $k = 0$ or $k = 2$.
We conclude by remarking that for the case of the harmonic oscillator, where $k=2$, the dissipated function \eqref{F0} has been already found in~\cite{bravetti2017contact}
by a direct calculation. 

\section{Harmonic type potentials with linear dissipation}
Clearly, Theorems~\ref{th:GNT2} and~\ref{th:IGNT2} do not apply only to scaling symmetries~\eqref{DefScal}.
To illustrate this point, let  
 us conclude by considering {one-dimensional} systems with Hamiltonians of the type 
	 \begin{equation}
	 h = \frac{1}{2m} p^{2} + \frac{m}{2} f(t) q^2 + g_0 S, \label{HamFin}
	 \end{equation}
\noindent where $f(t)$ is an arbitrary function of time $t$ and $g_0$ is an arbitrary real parameter. These systems can model
e.g.~a one dimensional harmonic oscillator with a linear dissipation and a time-dependent factor in the potential. 
{We recall that in this case, due to the explicit time dependence, the Hamiltonian $h$ is not a dissipated quantity.}

To derive the symmetries of this system let us write the dissipation equation \eqref{eq:Killing} explicitly 
	 \begin{equation}\label{eq:NewKilling}
	 \frac{p}{m} \frac{\partial F}{\partial q} - \left( m f(t) q + \lambda p \right) \frac{\partial F}{\partial p} + \left( \frac{1}{2m} p^2 
	 - \frac{m}{2} f(t) q^2 - g_0 S \right)  \frac{\partial F}{\partial S} +  \frac{\partial F}{\partial t} + g_0 F = 0\,.
	 \end{equation}
 
Let us look for generalized Noether symmetries of the Hamiltonian (\ref{HamFin}) for two particular cases: 
(i) the non-dissipative Hamiltonians, where $g_0 =0$, and 
(ii) linear dissipative term, where $g_0 \neq 0$. 
Each of these cases will be treated using the following ansatz for $F$: 
	 \begin{equation}
	 F(q,p, S, t) = A(q,S,t) p^2 + B(q,S,t) p + C(q,S,t)\,. \label{Ansatz}
	 \end{equation}
 
In the first case, when $g_0 =0$, the ansatz (\ref{Ansatz}) yields a solution of the form
	\begin{equation}
	F(q,p,S,t) =  C_0 F_0(q,p,S) + C_1 F_{LR}(q,p,t)\,,
	\end{equation}
\noindent where $F_0$ was defined in (\ref{F0})  and $F_{LR}$ is given by
	\begin{equation}\label{eq:LR}
	F_{LR}(q,p,t) = \frac{\rho(t)^2 }{2 m} p^2 - \rho(t) \dot{\rho}(t) q p + \frac{m}{2} \left( \dot{\rho}(t) + \frac{\rho_0}{\rho(t)^2} \right)^2 q^2.
	\end{equation}
Note that, since the dissipation equation is linear, we have obtained a linear combination of two independent solutions, $F_0$ and $F_{LR}$. 
Moreover, $F_{LR}$ is the well-known \emph{Lewis-Riesenfeld invariant}~\cite{lewis1968class,lewis1969exact}, and $\rho(t)$ has to satisfy the auxiliary 
Ermakov equation
	\begin{equation}
	\ddot{\rho}(t) + f(t) \rho(t) = \frac{\rho_0}{\rho(t)^3}\,, \label{AuxLR}
	\end{equation}
\noindent where $\rho_0$ is an arbitrary real constant. 
{It is not surprising that we have obtained the Lewis-Riesenfeld invariant as a generalized Noether invariant,
as it can be shown to be a standard Noether invariant~\cite{lutzky1978noether}. 
Moreover, we emphasize that both $F_0$ and $F_{LR}$ are conserved quantities in this case. 
}

{To generalize the above discussion, let} 
us consider now the case where $g_0 \neq 0$ and use the same ansatz \eqref{Ansatz}. 
A direct calculation shows 
that in this case the solution of the dissipation equation~\eqref{eq:NewKilling} takes the form
	\begin{equation}
	F(q,p,S,t) = C_0 F_0(q,p,S) + C_1 F_{GLR}(q,p,t)  + C_2 F_{EM}(q,p,t)
	\end{equation}
\noindent where $F_{GLR}$ and $F_{EM}$ are defined as
	\begin{equation}\label{eq:GLR}
	F_{GLR}(q,p,t) = a^2(t) \frac{p^2}{2 m} + \frac{1}{2} (g_0 a^2(t) -2 a(t) \dot{a}(t)) q p 
	+ \left[ \dot{a}^2(t) - g_0 a(t) \dot{a}(t) + \frac{g^2_0 a^2(t)}{4} + \frac{a^3_0}{a^2(t)} \left( 1 + \frac{3 a_0 g^2_0}{4}\right) \right] \frac{m q^2}{2}\,, 
	\end{equation}
	\begin{equation}
	F_{EM}(q,p,t) = b(t) p + m \dot{b}(t) q\, ,
	\end{equation}
\noindent and the auxiliary functions ${a(t)}$ and $b(t)$ satisfy the equations
\begin{eqnarray}
&& \ddot{a}(t) +  f(t) a(t) -  \frac{g^2_0}{4} a(t) = \frac{a^3_0}{a^3(t)} \left( 1 + \frac{3}{4} a_0 g^2_0\right), \label{AuxGLR}\\
&& \ddot{b}(t) + g_0 \dot{b}(t) +  f(t) b(t) = 0,
\end{eqnarray}
\noindent where $a_0$ is an arbitrary real constant.
 
The function $F_{GLR}$ can be considered as a generalization of the Lewis-Riesenfeld invariant $F_{LR}$ to the case where the system has a further linear dissipative term. 
This can be checked by taking $g_0=0$ in~\eqref{eq:GLR} and \eqref{AuxGLR} and observing that they reduce to \eqref{eq:LR} and \eqref{AuxLR} respectively. 
{Contrary to the previous case where $g_0 = 0$, these dissipated quantities are now not conserved. 
However, according to Corollary~\ref{cor:conserved2}, their quotient is a conserved quantity.}

To conclude this section, let us write the generalized Noether symmetry associated with the generalized Lewis-Riesenfeld dissipated quantity $F_{GLR}$, which reads
\begin{eqnarray}
X_{F_{GLR}} &=& \left[ \frac{a(t) p}{m} + \frac{1}{2} \left( g_0 a(t) - \dot{a}(t) \right) q \right] \frac{\partial}{\partial q} + \left[\frac{1}{2} \left( g_0 a(t) - \dot{a}(t) \right) p 
+ m \left( f(t) a(t) -\frac{g_0}{2} \dot{a}(t) + \frac{\ddot{a}(t)}{2} \right) q \right] \frac{\partial}{\partial p} + \nonumber  \\
&&  + \left[ \frac{a(t)  p^2}{2 m} - \frac{m}{2} \left(  f(t) a(t) -\frac{g_0}{2} \dot{a}(t) + \frac{\ddot{a}(t)}{2} \right) q^2 \right] \frac{\partial}{\partial S}\,,
\end{eqnarray}
where we are considering $Y^{t}=0$ in~\eqref{eq:YF}.

\section{Conclusions and future work}
We have proved a geometric extension of the generalized Noether theorem for scaling symmetries recently put forward in~\cite{zhang2020generalized}.
Our construction stems from the observation that the invariants associated with scaling symmetries, cf.~\eqref{eq:conservedZhang}, in general include an explicit 
dependence on the on-shell action of the system and on the time variable, 
and therefore we argued that the extended contact phase space is the appropriate minimal geometric setting to include such invariants and their related symmetries.
Indeed, by carefully constructing the extended contact phase space and the related Hamiltonian dynamics, we have shown that a sensible definition
of Noether symmetries exists such that an extension of the generalized Noether theorem for scaling symmetries that applies to all dynamical similarities and also
to some dissipative systems can be immediately found 
(Theorem~\ref{th:GNT2}),
together with its inverse statement (Theorem~\ref{th:IGNT2}).

As we have argued, these theorems have several positive features: in the first place, their proofs are very simple and are natural generalizations of their counterparts
in the standard symplectic and contact cases; moreover, they apply equivalently to conservative systems and to those dissipative systems that a admit a description in terms 
of time-dependent contact Hamiltonian systems; by construction, in this space the thus-obtained conserved or dissipated quantities are actual functions on the manifold and do not
contain spurious terms involving integrals over the dynamics; finally, one can directly show that the generalized Noether symmetries thus derived form a Lie algebra, and therefore they are amenable
to be treated in algebraic terms, in the lines of e.g.~\cite{marmo2020symmetries}.

We hope that the analysis initiated in this work will be helpful to clarify the origin and structure of scaling symmetries and their related invariants, by
putting them on the same footing as other standard Noether symmetries. 
As further developments, we will address the comparison of our results with the ``Eisenhart-Duval lift'' of mechanical systems
employed in~\cite{zhang2019kepler}, with the ``vertical extension'' of
the tangent space presented in~\cite{garcia2020vertical}, with the reduction in the pre-symplectic setting described in~\cite{marmo2020symmetries}, 
and also with the ``unit-free approach'' to Hamiltonian mechanics introduced in~\cite{zapata2020jacobi}, and we will provide
a deeper analysis of the Lie-algebraic structure of the generalized Noether symmetries for various systems of interest in physics, e.g.~in cosmology~\cite{sloan2018dynamical,sloan2019scalar}.



\end{document}